%% file: Revision-ElsTemplate_signalling_for_demand_response.tex
\documentclass[preprint,fleqn,5p,twocolumn]{elsarticle}
\usepackage[T1]{fontenc}
\usepackage[utf8]{inputenc}
\usepackage[english]{babel}

\usepackage{amssymb}
\usepackage[fleqn]{amsmath}
\def \E{\mathbb{E}}

\def \H{\mathbb{H}}
\def \L{\mathbb{L}}
\def \M{\mathbb{M}}
\def \N{\mathbb{N}}
\def \P{\mathbb{P}}
\def \Q{\mathbb{Q}}
\def \R{\mathbb{R}}

\def \Z{\mathbb{Z}}

\def \G{\mathbb{G}}

\def\={\;=\;}
\def\.{\;.}

\def \i{1\!\mbox{\rm I}}
\def\1{{\bf 1}}

\def\normeL2#1{\left\|{#1}\right\|_{L^2}}

\newcommand{\alias}[2]{
	\providecommand{#1}{}
	\renewcommand{#1}{#2}
}

\alias{\N}{\mathcal{N}}
\alias{\L}{\mathcal{L}}
\alias{\Z}{\mathbb{Z}}
\alias{\Q}{\mathbb{Q}}
\alias{\R}{\mathbb{R}}
\alias{\C}{\mathcal{C}}
\alias{\T}{\mathbb{T}}
\alias{\E}{\mathbb{E}}
\alias{\H}{\mathcal{H}}
\alias{\1}{\mathbf{1}}
\alias{\B}{\mathcal{B}}
\alias{\M}{\mathcal{M}}
\alias{\G}{\mathcal{G}}
\alias{\Y}{Y_{\bullet}}

\newcommand{\nc}{\newcommand}
\nc{\cA}{{\mathcal A}} \nc{\cB}{{\mathcal B}} \nc{\cC}{{\mathcal
		C}} \nc{\cD}{{\mathcal D}} \nc{\bbD}{\mathbb{D}}
\nc{\cG}{{\mathcal G}} \nc{\cF}{{\mathcal F}} \nc{\cS}{{\mathcal
		S}} \nc{\cU}{{\mathcal U}} \nc{\cH}{{\mathcal H}}
\nc{\cK}{{\mathcal K}} \nc{\cM}{{\mathcal M}} \nc{\cO}{{\mathcal
		O}} \nc{\cP}{{\mathcal P}} \nc{\bbE}{\mathbb{E}}
\nc{\bbEP}{\mathbb{E}_{\mathbb{P}}}\nc{\bbL}{\mathbb{L}}
\nc{\bbP}{\mathbb{P}} \nc{\bbQ}{\mathbb{Q}} \nc{\del}{\partial}
\nc{\Om}{\Omega} \nc{\om}{\omega} \nc{\bbR}{\mathbb{R}}
\nc{\bbC}{\mathbb{C}} \nc{\bfr}{\begin{flushright}}
	\nc{\efr}{\end{flushright}} \nc{\dXt}{\delta q_{t}}
\nc{\dXs}{\delta q_{s}} \nc{\bs}{\blacksquare} \nc{\dX}{\delta q}
\nc{\dY}{\Delta Y}
\nc{\dnkx}{\left(X(T^{n}_{k})-X(T^{n}_{k-1})\right)}
\nc{\esssup}{\mathrm{ess}\mbox{ }\mathrm{sup}}
\nc{\essinf}{\mathrm{ess}\mbox{ } \mathrm{inf}}
\nc{\dhats}{\widehat{\delta_s}}

\nc{\chf}{\mbox{$\mathbf1$}}
\nc{\ind}{\mathds{1}}

\begingroup\expandafter\expandafter\expandafter\endgroup
\expandafter\ifx\csname pdfsuppresswarningpagegroup\endcsname\relax
\else
\fi
\graphicspath{{./fig/}}

\nc{\mum}{ \mu_{\rm m} }
\nc{\muv}{ \mu_{\rm v} }
\nc{\mumv}{ \mu_{\rm mv} }
\nc{\Hm}{ H_{\rm m} }
\nc{\Hv}{ H_{\rm v} }

\newcommand{\myblue}[1]{\textcolor{black}{#1}}

\input{macros.tex}

\input{macros-math.tex}
\alias{\P}{\mathbb{P}}

\def \ud{\mathrm{d}}

\begin{document}

\begin{frontmatter}



\title{\LARGE \textbf{Signalling for Electricity Demand Response: \\When is Truth Telling Optimal?}}


\author[1]{Ren{\'e} A{\"\i}d}
\author[2]{Anupama Kowli}
\author[3]{Ankur A. Kulkarni}
\address[1]{University Paris Dauphine --- PSL Research University --- LEDa UMR 8007-260, {Paris}, {France}, {rene.aid@dauphine.psl.eu}}
\address[2]{{Electrical Engineering, Indian Institute of Technology Bombay},
	{Powai},
	{Mumbai},
	{400076},
	{Maharashtra},
	{India. anu.kowli@iitb.ac.in}}
\address[3]{{Systems and Control Engineering, Indian Institute of Technology Bombay},
	{Powai},
	{Mumbai},
	{400076},
	{Maharashtra},
	{India. kulkarni.ankur@iitb.ac.in}}
\begin{abstract}
Utilities and transmission system operators (TSO) around the world implement demand response programs for reducing electricity consumption by sending information on the state of balance between supply demand to end-use consumers. We construct a Bayesian persuasion model to analyse such demand response programs. Using a simple model consisting of two time steps for contract signing and invoking, we analyse the relation between the pricing of electricity and the incentives of the TSO to garble information about the true state of the generation. We show that if the electricity is priced at its marginal cost of production, the TSO has no incentive to lie and always tells the truth. On the other hand, we provide conditions where overpricing of electricity leads the TSO to provide no information to the consumer.
\end{abstract}

%

\begin{keyword}
Bayesian persuasion, demand response, pricing, signalling
\end{keyword}

\end{frontmatter}

\section{Introduction}
\input{introDR.tex}
\input{bayesianpersuasion-RA.tex}

\input{model.tex}
\input{results.tex}
\input{proofs.tex}

\bibliographystyle{plain}

\end{document}

%% file: macros.tex
\usepackage{amsmath, amssymb, xspace,url}
\usepackage{epsfig}
\usepackage{longtable}

\usepackage[pdftex,dvipsnames]{xcolor}
\usepackage{todonotes}
\usepackage[framemethod=tikz]{mdframed}
\usepackage{lineno}
\newmdenv[leftmargin=\dimexpr-0.4em, innerleftmargin=0.5em,
rightmargin=\dimexpr-0.4em, innerrightmargin=0.5em,
linewidth=2pt,linecolor=red, topline=false, bottomline=false,
innertopmargin=0pt,innerbottommargin=0pt,skipbelow=0pt,skipabove=0pt,%
]{notex}

\newenvironment{note}%
{\vskip\dimexpr\dp\strutbox-\prevdepth\relax\notex\strut\ignorespaces}%
{\xdef\notetpd{\the\prevdepth}\endnotex\vskip-\notetpd\relax}

\let\oldtodo\todo

\makeatletter%
\DeclareDocumentCommand{\todo}{ O{} +g +d<> }{%
	\IfNoValueTF{#2}{\relax}{%
		\oldtodo[caption={#2},size=\footnotesize,#1]{\renewcommand{\baselinestretch}{1}\selectfont\sffamily#2\par}%
	}%
	\IfNoValueTF{#3}{\relax}{%
		\IfNoValueTF{#2}{
			\begin{note}%
				\begin{internallinenumbers}%
					\indent%
					#3%
				\end{internallinenumbers}%
			\end{note}%
		}{
			\vspace{-0\baselineskip}%
			\begin{note}%
				\begin{internallinenumbers}%
					\indent%
					#3%
				\end{internallinenumbers}%
			\end{note}%
		}%
	}%
}%
\makeatother
\usepackage{soul}

\def\bkE{{\rm I\kern-.17em E}}
\def\bk1{{\rm 1\kern-.17em l}}
\def\bkD{{\rm I\kern-.17em D}}
\def\bkR{{\rm I\kern-.17em R}}
\def\bkP{{\rm I\kern-.17em P}}

\def\bkZ{{\bf{Z}}}
\def\bkE{{\rm I\kern-.17em E}}
\def\bk1{{\rm 1\kern-.17em l}}
\def\bkD{{\rm I\kern-.17em D}}
\def\bkR{{\rm I\kern-.17em R}}
\def\bkP{{\rm I\kern-.17em P}}
		\def\bkE{{\rm I\kern-.17em E}}
		\def\bk1{{\rm 1\kern-.17em l}}
		\def\bkD{{\rm I\kern-.17em D}}
		\def\bkR{{\rm I\kern-.17em R}}
		\def\bkP{{\rm I\kern-.17em P}}
		\def\bkY{{\bf \kern-.17em Y}}
		\def\bkZ{{\bf \kern-.17em Z}}
		\def\bkC{{\bf  \kern-.17em C}}

\def\del{\partial}

\let\forallnew\forall
\renewcommand{\forall}{\forallnew\ }
\let\forall\forallnew

\def\b12{(\beta_1,\beta_2)}

\newcommand{\Real}{\ensuremath{\mathbb{R}}}

\def\ind{{\rm ind}}

\def\hbar{\skew{4.2}\bar h}

\def\L{_{\scriptscriptstyle L}}
\def\M{_{\scriptscriptstyle M}}
\def\N{_{\scriptscriptstyle N}}

\def\Q{_{\scriptscriptstyle Q}}
\def\P{_{\scriptscriptstyle P}}
\def\R{_{\scriptscriptstyle R}}

\def\Y{_{\scriptscriptstyle Y}}
\def\Z{_{\scriptscriptstyle Z}}

\def\superstar{^{\raise 0.5pt\hbox{$\nthinsp *$}}}
\def\SUPERSTAR{^{\raise 0.5pt\hbox{$*$}}}

\def\lamstarT {\lambda^{\raise 0.5pt\hbox{$\nthinsp *$}T}}

\def\del{\partial}

\newlength{\noteWidth}
\long\def\notes#1{\ifinner
{\tiny #1}
\else
\marginpar{\parbox[t]{\noteWidth}{\raggedright\tiny #1}}
\fi\typeout{#1}}
 \def\notes#1{\typeout{read notes: #1}} 

\newcommand{\ie}{i.e.\@\xspace} 
\newcommand{\eg}{e.g.\@\xspace} 

\def\spose#1{\hbox to 0pt{#1\hss}}

\def\text #1{\hbox{\quad#1\quad}}

\def\T{^T\!}

\makeatletter
\newcommand{\pushright}[1]{\ifmeasuring@#1\else\omit\hfill$\displaystyle#1$\fi\ignorespaces}
\newcommand{\pushleft}[1]{\ifmeasuring@#1\else\omit$\displaystyle#1$\hfill\fi\ignorespaces}
\makeatother

\def\nthinsp{\mskip -2   mu}

\usepackage{ifthen}
\newboolean{showcomments}
\setboolean{showcomments}{true}
\newcommand{\aak}[1]{  \ifthenelse{\boolean{showcomments}}
{ \textcolor{blue}{(AAK says:  #1)}} {}  }

%% file: macros-math.tex
\newtheorem{theorem}{Theorem}[section]

\newtheorem{proposition}[theorem]{Proposition}

\newcounter{example}
\renewcommand{\theexample}{\thesection.\arabic{example}}

\newcounter{remark}
\renewcommand{\theremark}{\thesection.\arabic{remark}}
\newenvironment{remarkc}[1][]{\refstepcounter{remark}
\par\medskip \noindent%
   \textbf{Remark~\theremark. #1} \rmfamily}{\hfill $\square$   \hspace{-4.5pt} \vspace{6pt}}

{\begin{list}{}%
         {\setlength{\leftmargin}{#1}}%
         \item[]%
}
{\end{list}}

		\def\bsp{\begin{split}}
		\def\beq{\begin{eqnarray}}
		\def\bal{\begin{align*}}
		\def\bc{\begin{center}}
		\def\be{\begin{enumerate}}
		\def\bi{\begin{itemize}}
		\def\bs{\begin{small}}
		\def\bS{\begin{slide}}
		\def\ec{\end{center}}
		\def\ee{\end{enumerate}}
		\def\ei{\end{itemize}}
		\def\es{\end{small}}
		\def\eS{\end{slide}}
		\def\eeq{\end{eqnarray}}
		\def\eal{\end{align*}}
		\def\esp{\end{split}}
		\def\qed{ \vrule height7.5pt width7.5pt depth0pt}  

	\def\cp2problem#1#2#3#4{\fbox
		 {\begin{tabular*}{0.9\textwidth}
			{@{}l@{\extracolsep{\fill}}l@{\extracolsep{6pt}}l@{\extracolsep{\fill}}c@{}}
				#1 & & $#4 $ 
			\end{tabular*}}}

		\renewcommand{\emph}[1]{\textbf{#1}}

		\def\bkE{{\rm I\kern-.17em E}}
		\def\bk1{{\rm 1\kern-.17em l}}
		\def\bkD{{\rm I\kern-.17em D}}
		\def\bkR{{\rm I\kern-.17em R}}
		\def\bkP{{\rm I\kern-.17em P}}
		
		\def\bkZ{{\bf{Z}}}

\newcommand {\beeq}[1]{\begin{equation}\label{#1}}
\newcommand {\eeeq}{\end{equation}}

\def\texitem#1{\par\smallskip\noindent\hangindent 25pt
               \hbox to 25pt {\hss #1 ~}\ignorespaces}



%% file: introDR.tex
The last decade has seen significant growth and deployment of renewable energy resources in the power sector \cite{iea2021renewables}. Since renewable power generation is inherently uncertain, variable and hard to control, there is a need to explore sources of flexibility such as energy storage, flexible loads and controllable generators that can help mitigate the effects of renewable power. In particular, last decade has seen several demand response programs where load curtailments help balance out supply deficits. 
These curtailments can be mandated, or voluntary as a response to price/tariff signals, control signals or other means of signals. 
For instance, utilities may have direct load control contracts with end users or may impose peak pricing tariffs on them to induce load curtailments and/or shifting. 
Such demand response schemes are known to lower electricity prices and electricity production costs \cite{energy2022quarterly}. In this paper, we are interested in analyzing the outcomes of demand response programs where participation is voluntary and consumers react to signals received from the coordinator for DR. 

The setting considered in this paper is relevant in situations where demand response is induced using alternative signals other than pricing. For instance, the {\em MonEcoWatt} program of the French TSO,  RTE~\footnote{\url{https://www.monecowatt.fr}, only in French} signals the state of the electric systems to induce a response. This is achieved by sending phone text messages when the system is stressed to voluntarily enrolled consumers who receive no economic compensation for their efforts, if any. The signal takes the form of a color: green when no need of action is required, yellow when the load is high, orange when reduction actions would be welcomed and red when they are mostly necessary to avoid load shedding.  The French TSO has measured that orange alerts reduced the peak-load consumption by 3\% in the region of Britain. \myblue{There are similar other programs such as {\em Shave The Peak}\footnote{https://www.greenenergyconsumers.org/shavethepeak} program in MA, USA and the {\em Power Alert}\footnote{https://www.eskom.co.za/eas/power-alert/} program in South Africa, where signals are sent out to enrolled users to indicate the grid stress. The signal can be binary -- stress or no stress -- or coloured where the color indicates the level of stress.}
There is not much academic literature on understanding the impacts of such signalling on the actions of the end users. 

%
It is plausible that the impacts of such signalling would also depend on the pricing regime governing the consumer's behaviour. Signalling and pricing could create conflicting incentives, or create overlapping effects. Again, there does not appear to be any literature covering this matter. We seek to fill this gap with a formal model of signalling in such settings.

\subsection{Main contributions}
We are concerned with a setting in which a TSO and a consumer are bound by a pre-agreed pricing contract that is in force for the duration being considered here.  This contract is designed taking into account the statistics of the generation known at the time of signing the contract. However, owing to the large time duration in the contract, scenarios could arise wherein the contract is valid but not optimal any more given changed circumstances. In addition, in real time specific values of the uncertainty get realized due to natural fluctuations and the TSO may want to take advantage of this information. While the pricing contract cannot be modified on a scenario-wise basis, we explore the possibility for the TSO to influence the consumer's action by sending signals to persuade the consumer to act accordingly. The design of such signals and their truthfulness is the subject of this paper.


We model a setting in which a TSO and a consumer, share, in addition to a binding pricing contract, a common prior distribution on the level of future generation. However, the realized level of generation is known only to the TSO. The TSO must decide how it may use this fresh information to influence the consumer. Decision-theoretically speaking, the consumer is the only player that can act in our setting: the amount of consumption chosen by the consumer determines the payoff of both players. Thus the TSO's decisions are limited to only influencing this choice of the consumer by reshaping the latter's information by sending suitable signals.

Formally we model our setting as that of Bayesian persuasion with the TSO as the sender and the consumer as the receiver of information about a state. In such a setting, the sender sends signals\footnote{We assume a setting of costless signalling, popularly known as \textit{cheap talk}.} and the receiver updates its belief about the unknown state using Bayes rule. The problem for the TSO is to design signals such that the TSO obtains the maximum payoff under the optimal action chosen by the consumer under the posterior distribution. 

Our quest is to illuminate the role of truthtelling on the part of the TSO. More specifically, when do the signals sent by the TSO constitute truth, and nothing but the truth? Our main finding is that there is a close link between the nature of the contract and truthtelling -- if the contract is such that the price equals the marginal cost of generation, truthtelling is optimal. In other words, if the contract is ``priced to perfection'' in that it has accurately estimated the probability distribution of the future generation, there are no additional benefits to be drawn from reshaping the truth. Similar results hold even in the case where there are multiple consumers. 

Fascinatingly ``mispriced'' contracts may also induce truthtelling -- this happens to be the case when the mispricing is in a direction where telling the truth is beneficial to the TSO. When the mispricing is in the opposite direction, garbling  the truth is optimal.

\subsection{Related works}
There is a large academic literature on demand response in electricity from the lens of \textit{mechanism design}~\cite{li2017mechanism,muthirayan2019mechanism,muthirayan2016mechanism,nekouei2014game,cao2012optimal,Fahrioglu00}. These works focus on pricing or other mechanism for influencing behaviour, or for eliciting information from consumers regarding their types (\eg, \cite{muthirayan2016mechanism} focuses on reporting baselines). Our approach is in a sense \textit{dual} to this -- in our case the TSO \textit{provides} information (possibly garbled) to the consumer.

Consequently, in our analysis, we ignore two features -- adverse selection and moral hazard -- of incentives theory. 
Adverse selection is relevant to settings where there are multiple types of consumers and the TSO is ignorant of the type of the consumer but knows the distribution of types in the population. In such cases, the pricing rule may depend on the optimal reduction effort as shown in \cite{Fahrioglu00} and \cite{Crampes15}. We circumvent the adverse selection problem by restricting our analysis to only one type of consumer. 
The moral hazard problem  pertains to the settings where the TSO cannot distinguish when the load curtailment is the result of the action of the consumer and when it is purely the outcome of a random event. In such a situation, the pricing rule has to take into account this feature (see \cite{aid2022optimal}). We simplify our analysis by restricting our setting to cases where the consumption is a deterministic function of the consumer's action. 

There at least four key differences between the instruments of long-term pricing and real-time signalling -- the timescales at which these are implemented, the knowledge they assume of the uncertainty, the costs involved in implementing them, and the compliance they induce. 
Incentives such as pricing rules and payment terms, once decided are held fixed for long  periods of time. As such they are usually updated on a time frame of a few months or even years. Signals, on the other hand, can be sent to the consumer on a much more frequent basis. Indeed, the TSO may opt to send a fresh signal whenever it has new information, such as updated weather forecasts, any maintenance related down times and so on. The pricing contract is designed before the realization of uncertainty whereas signals are adapted to realized uncertainty. In addition,  there is also a difference in the costs involved. Contracts entail legal, regulatory and marketing costs, whereas in our model, signals can be sent virtually free of cost via notifications on mobile phones or SMSes as the MonEcoWatt program. Finally, contracts are of a legally binding nature and therefore ensure compliance. On the other hand signals can be ignored by the consumer.

%% file: bayesianpersuasion-RA.tex
In a sender-receiver interaction~\cite{myerson1997game}, the setting where the sender shapes the probability distribution of the receiver by sending suitable signals is called \textit{Bayesian persuasion.}
In this setting, the sender and receiver have a prior distribution on a payoff-relevant state, whose realized value is known only to the sender. The sender sends signals, upon receiving which, the receiver computes a posterior distribution on the state using Bayes rule. The receiver then takes an action that maximizes the expected utility computed with this posterior distribution. In the simplest signalling setting, the sender has no action of its own, and the receiver's action determines the payoff both players. 
The goal of the sender is to design a signalling policy anticipating this response of the sender.

The persuasion setting falls broadly in the area of \textit{information design}~\cite{kamenica2018bayesian}. The use of randomization by the sender leads to the receiver getting `mixed signals'; the result being that the true type of the sender gets obfuscated, and a benefit accrues to the sender. This deliberate, strategic ambiguity is a hallmark of such settings, see~\cite{kamenica2011bayesian}, \cite{forges2020annals} and \cite{forges2008journal}. Having said that, there are also models where randomization is \textit{not} employed, and yet qualitatively similar phenomena occur (see e.g.,~\cite{deori2022information}). 
Qualitatively similar behaviour has also been seen in the \textit{screening} setting~\cite{rasmusen_games_2006} where the receiver commits first, (see e.g., \cite{vora2021optimal} and \cite{vora2021information}),
and where players move simultaneously~\cite{crawford1982strategic}. 

Considering that the signalling generally results in the sender obfuscating the truth, it instructive to ask when is truth telling optimal for the sender? Distinct classes of equilibria called separating, pooling and semi-separating equilibria, have been studied~\cite{sobel2009signaling} in this context. In a separating equilibrium, the receiver is able to discern the exact type from the sender's signal, whereas in a pooling equilibrium, all sender types employ the same signal. An equilibrium that is neither separating nor pooling is called semi-separating. Recent work~\cite{deori2022information} characterizes these equilibria in a non-Bayesian setting through a quantity called the \textit{informativeness} of the sender, in a specific setting where the receiver is interested in knowing the sender's true type.
Our contribution in this paper lies in relating truth-telling to the pricing of a pre-existing  contract, and, in particular, to estimation of future uncertainties. To the best of our knowledge this connection has not been established before, more so in the context of demand response.

%% file: model.tex
\section{Model}

\subsection{Model: single consumer}

We consider an electricity market where the TSO acts as single buyer of all electricity generation from conventional producers and renewable energy producers. We assume the generation is served by enough producers so that the market can be considered competitive. We assume that the TSO has a perfect forecast of the renewable generation denoted $ q $, but this level is not known to the consumers. \myblue{Indeed, even though forecasting renewable energy generation at a daily-time scale is yet an imprecise exercise, the knowlegde of the consumers are significantly lower compared to TSO's.} In this section we consider that consumers can be aggregated in a single representative consumer whose current consumption is denoted $ y_0 $.

The consumer can take an action $ a  \in \Real $ to modify its consumption. \myblue{To focus on the effects of the information known to the TSO, we assume that there is no noise in the action of the consumer whereby the consumption after action is given as $y = y_0 - a$. Thus, by observing the consumption, the TSO knows the action of the consumer and is not facing an observation problem as in the moral hazard demand-response model in \cite{aid2022optimal}.} The consumer's decision $a$ has to be taken prior to the knowledge of the true level of renewable generation $q \in \Real$. The consumer and the TSO are bound by a pricing contract where the per unit price charged for consumption is a function $ p(y, q) $ of the consumption $ y $ and the generation $ q $. We assume this function takes the following form
$$
p(y,q) = p(y_0 - q - a),
$$
whereby the unit price is a function difference between the consumption and the generation. Note that, though the generation $ q $ is not known to the consumer, the above pricing rule is fixed and known.
We define the utility of the consumer by
\begin{equation}
\begin{split}
	U(q, a) :=	&  \quad u(y) - c(a)  - p(y-q-a) y,
\end{split}
\end{equation}
where $u(y)$ is the indirect utility of consuming electricity $y$, $c(a)$ is the cost of making undesired modification of consumption $a$. We take the instanciation
$$u(y) := u y\  \text{and}\  c(a) := \tfrac{1}{2} c a^2,$$
where $u, c>0.$

\myblue{The TSO is assumed to know the utility function of the representative consumer.}

We assume that both the TSO and the consumer share the same prior on $q$ denoted by $\P^0$, and under this prior $q$ has bounded support and finite variance. This prior can be thought as the known statistical distribution of wind during the year.  \myblue{The TSO observes the generation $ q $ and sends a message to the consumer.}

The \textit{signalling policy} of the TSO specifies a probability distribution on the messages for each value of $ q $. The policy is denoted $\pi( \tilde q | q )$, where $\tilde q$ is the message sent and $q$  the observed renewable energy generation. We assume that the signal and the renewable generation share the same bounded support. Telling the truth, all the truth and nothing but the truth consists in setting $\pi(\tilde q | q) = 1$ iff $\tilde q = q$.

\begin{remarkc}[Cheap talk:]
	We assume that signals can be sent free of cost, commonly known as the setting of \textit{cheap talk}.
	In other words, the cost of acquiring information, computing the signal and that of transmitting it is negligible. Physically, our assumption can be justified by observing that for the TSO all these costs are invariant with the actual signal being sent. Weather probes would typical incur a fixed cost and a constant operating cost irrespective of the weather itself. Computation and transmission costs have fallen dramatically over the last two decades and can be assumed to be negligible.
\end{remarkc}

\myblue{\begin{remarkc}[Committment:]
	We assume that the TSO commits to a signalling policy. Indeed, if the distribution of the renewable energy is stationary and if the pricing rule does not change, the TSO has no incentive to deviate from an optimal signalling policy.
\end{remarkc}}

\myblue{\begin{remarkc}[Pricing rule:]
We assumed that the pricing rule was defined prior to the design of information provision. One could object that the regulator might also act as player in the game between the TSO and the consumers which would make the pricing rule endogenous, provided a suitable objective function of the regulator was given. Yet, this setting might not completely reflect some electricity markets design where retailers might themselves develop demand-response contracts. Thus, the interaction between the pricing of electricity contracts to consumers and the information provision by TSOs is a problem beyond this model. For that reason, we postpone to further research work the analysis of these situations.
\end{remarkc}}

For a given signalling policy $\pi$ and for a given signal $ \tilde{q} $, the consumer's problem is to solve:
\begin{align}
	\sup_{a} J(a; \tilde q, \pi) := \E\big[U(q, a) | \tilde{q} \big]. \label{eq:consumerprob}
\end{align}
Assuming this problem admits a unique solution, we denote the best response of the consumer by $\hat a^\pi(\tilde q)$.

Knowing the best-responses $\hat a^\pi$ of the consumer to a given signalling policy $\pi$, the TSO chooses the signalling policy to solve,
\begin{align} \label{eq:TSO}
\inf_{\pi} J^{{\rm T}}(\pi) := &\E \big[ f( y_0 - \hat a^\pi(\tilde q) - q) \\
		&- p(y_0 - \hat a^\pi(\tilde q) - q)\big(y_0 - \hat a^\pi(\tilde q)\big) \big], \nonumber
\end{align}
where the expectation is taken under  $\pi \otimes \P^0$ for the realisation of the pair  $(\tilde{q}, q)  $ and $ f(x) $ is the cost that the TSO has to bear due to the mismatch $x:= y_0 - \hat a^\pi(\tilde q) - q$ between the renewable generation and the consumption. This may be interpreted as the cost of procuring conventional generation in the event of a shortfall in renewable generation. Note that $x<0$ can be interpreted as the activation by the TSO of contractual demand-response. We assume a linear quadratic production cost function, i.e.,
\begin{align}
	f(x) := k x + \frac12 \beta x^2, \quad 0 < k, \quad 0 < \beta,
\end{align}
so that the marginal cost curve is affine in production. Although in real life, the marginal cost of electricity production exhibits a steeper curvature, the assumption we make is common in the economic literature of electricity market regulation as it allows closed-form expressions for comparative static studies (see examples in \cite{aid2011} and \cite{Crampes15}). Besides, linear-quadratic preferences has also been investigated in the Bayesian persuasion economic literature as an important case (see \cite{tamura2018bayesian}).

The pricing rule is defined as the marginal cost of production of electricity, namely
\begin{align}
p(x) 	= f'(x) =  k  +  \beta x.
\end{align}
The significance and implications of this assumption will be discussed in the following section. We will later consider deviations from this pricing rule. \myblue{But note that taking the pricing rule as the marginal cost of production induces an effort reduction of the consumer that precisely corresponds to social welfare maximisation in the case of full information on the renewable generation as shown in Boiteux's marginal cost pricing of electricity theory (\cite{Boiteux51,Boiteux60})}.

We now reflect on how the consumer approaches this problem. To solve his  problem \eqref{eq:consumerprob}, the consumer computes the posterior probability that the true generation is $q$, upon receiving the message $\tilde q$:
\begin{align}
	&\P(  q | \tilde q ) =   \frac{\pi( \tilde q | q ) \P^0(q)}{  \E^{\P^0}\big[ \pi( \tilde q | q) \big]  }, \\
	&  \text{with}
	\E^{\P^0}\big[ \pi( \tilde q | q) \big]\!\!  = \!\!\int \!\! \pi(\tilde q | q ) \ud\P^0(q).
\end{align}
assuming that the denominator above is nonzero. Thus the expectation in \eqref{eq:consumerprob} is taken with respect to the posterior distribution above defined by $ \P(  q | \tilde q ). $

%% file: results.tex
\section{Results}

\subsection{When is truth telling is optimal?}
We now present the first main contribution of this paper. Recall that we have assumed that the TSO and the consumer have agreed on a pricing contract.

Recall also that we had assumed that the unit price was equal to the marginal cost of generation. This assumption implies that the cost of generation is known perfectly and this cost has been \textit{priced in} in the pricing formula. However, the pricing formula states that the price is a function of the difference between the generation $ q $ and the consumption $ y.$ Since $ q $ is not known to the consumer, a natural question is whether the TSO should convey some information to about $ q $ to the consumer. And if so, what should be the information it should convey?

	Our first main result shows that when the pricing is ``perfect'' in a certain sense, the optimal signal for the TSO is to disclose $ q $ without any ambiguity. In other words, truth telling is optimal for the TSO.

We define the following constants
\begin{align*}
a_0 := \frac{k+2\beta y_0 - u}{2\beta+c}, \quad  \lambda := \frac{\beta}{2\beta+c}
\end{align*}

\begin{proposition} \label{prop:truth} It holds that
\begin{enumerate}[\upshape (i)]
	\item The best-response function of the consumer $\hat a^{\pi}$ for a given signalling policy $\pi$  is given by
\begin{align}
\hat a^{\pi} = a_0 - \lambda \hat q, \quad  \hat q := \E\big[q | \tilde q\big].
\end{align}
As a result, $\E\big[ \hat a^\pi \big] = a_0 - \lambda \E\big[ q \big]$  does not depend on~$\pi$.
	\item  It is optimal for the TSO to tell the truth.
	\end{enumerate}
\end{proposition}
\begin{proof}
{\rm (i)}  For a given signal $\pi$, consider the optimisation problem of the consumer given by  \eqref{eq:consumerprob}.
We have  $J(a,\tilde q) = \E^{\P}\big[ U(q,a) \big].$ Taking derivative inside the expectation and writing first-order condition provides
\begin{align*}
\E^{\P}\big[ (c+2\beta) a - (c+2\beta) a_0 + \beta q \big] =0.
\end{align*}
Thus the optimal response of the consumer is given by $\hat a^\pi =a_0 - \lambda \E^{\P}[q]$ and $ \E^{\P}[q] = \E[q|\tilde q],$ which gives the result. Besides, by the law of iterated expectations, we get $\E\big[ \hat a^\pi \big]  = a_0 - \lambda \E\big[q\big].$

\noindent {\rm (ii)}  The criterion of the TSO can be written as
\begin{align*}
	& J^{\rm T}(\pi) = \E\big[ F(q) \big]  +  \beta y_0  \E\big[ \hat a^\pi \big] - \frac12 \beta \E\big[ (\hat a^\pi)^2 \big],
\end{align*}
with $ F(q)\!:=\!\frac12 \beta (y_0-q)^2\!- kq\!- \!\beta (y_0-q) y_0$. The two first terms are independent of $\pi$. Besides, because of~(i), it holds that $\hat a^\pi = a_0 - \lambda \hat q$ and thus, $\E[ (\hat a^\pi)^2] = a_0^2 -2a_0\lambda   \E[\hat q] + \lambda^2 \E[ \hat q^2]$. Because $\E[\hat q] = \E^{\P}[\hat q] = \E^{\P^0}[q]$, the only term that depends on $\pi$ in the criteria of the TSO is
$$- \frac12 \beta \lambda^2 \E\big[ \hat q^2 \big].$$ Thus, minimizing $J^{\rm T}(\pi)$ is equivalent to maximizing $\E[\hat q^2]$. Now, for any random variable $m$ correlated with $q$, the following inequalities hold from Jensen's inequality:
\begin{align}\label{eq:ineq}
\E[q]^2 \!\! = \!\E[ \E[ q|m] ]^2 \leq \E[ \E[ q|m]^2  ] \leq \E[ \E[ q^2 | m] ] = \E[ q^2].
\end{align}
When $m$ is independent of $q$ (\ie, no information is revealed),  $\E[\hat q^2] = \E[q]^2$ which is the minimum, whereas if $m=q$ (full information revelation),  $\E[\hat q^2] = \E[q^2]$ which the maximum. Thus, the TSO's criteria is minimized when he tells the truth.\hfill$\Box$
\end{proof}

\myblue{\begin{remarkc}[Increasing consumption:]
Even if what can be considered as the baseline consumption level $a_0$ is positive, it might happen that $\hat a^\pi(\tilde q)<0$. When $c$ is small, the consumer is responsive to the information on the availability of cheap energy and takes these opportunities to increase his consumption. Although our model design is not intented to capture this phenomenon, it is a behaviour looked for by TSOs. The Low Carbon London research project is an example of a large scale demand-response research program where incentives to increase consumption were proposed to consumers \cite{Schofield14}).
\end{remarkc}}

\myblue{\begin{remarkc}[Trustful consumer:]
	 Assume the consumer does not behave as a Bayesian rational decision-maker but instead fully trust the sender's signal similar to the setting considered in \cite{kulkarni2019addressing}. In that situation, the optimal response of the consumer is $\hat a^\pi(\tilde q)  = a_0 - \lambda \tilde q,$ because he would strictly believe that the signal is the truth and makes his optimisation accordingly. Hence, his actions are independent of $\pi$ and $\E[\hat a^\pi] = a_0 - \lambda \E[\tilde q]$. There are no reasons why $\E[\tilde q]$ would be equal to $\E[q]$. Redoing the computation of the result~(ii) in the Proof above for the criteria of the TSO shows that  $J^{\rm T}$ is a constant plus the term
\begin{align*}
&-\frac12 \beta \lambda^2  \E\Big[ \Big( \tilde q - \frac{a_0-y_0}{\lambda} \Big)^2 \Big].
\end{align*}
Thus, the TSO should maximize the $L^2$ distance of the signal from $\frac{a_0-y_0}{\lambda}$. Clearly, taking $\tilde{q}$ to be a constant equal to the end point the support of $q$ farthest from $\frac{a_0-y_0}{\lambda}$ maximizes the TSO's criterion.
\end{remarkc}}

Proposition~\ref{prop:truth} states that the marginal \textit{cost} pricing rule ensures that the TSO has no incentive  to garble the production forecast and that the consumer's optimal response under optimal signalling is the social optimum.  The most extreme case of garbling information would be if the TSO gives an uninformative signal. In other words, the posterior probability is equal to the prior. In that case the TSO's payoff is minimum. Thus, sending any signal is better than sending no signal at all and sending a truthful signal is the optimal policy.

\subsection{Deviation from marginal cost pricing}
We investigate now if this property still applies if we deviate from the marginal production cost pricing rule.  We assume now that the pricing function $p$ is given by
\begin{align}
	p(x) = k + b x, \quad b \neq \beta.
\end{align}
Thus,  the TSO charges a price either with a higher or lower slope compared to $\beta$. In that situation, the only change in the best-response of the consumer is that $\beta$ is replaced by $b$. Thus, we have
\begin{align*}
\hat a^\pi = a_{0,b} - \lambda_b \hat q, \quad a_{0,b} := \frac{k +2by_t - u}{c+2b}, \quad \lambda_b := \frac{b}{c+2b}.
\end{align*}
The criteria of the TSO now becomes
\begin{align*}
	& J^{\rm T}(\pi) = \E\big[ F(q) \big] + (2b - \beta) y_0  \E\big[ \hat a^\pi \big] \\
	&+ (\beta - b) \E\big[ q \hat a^\pi\big] + (\frac12 \beta - b) \E\big[(\hat a^\pi)^2\big].
\end{align*}
Besides, we have $\E^{\pi \otimes \P}\big[ q \hat q \big] = \E^{\pi \otimes \P}\big[  (\hat q)^2 \big].$
Thus, the TSO's criteria is a constant term plus
\begin{align*}
	& \lambda_b \Big[ \underbrace{ - (\beta - b) + \lambda_b (\frac12 \beta - b)}_{=:\Lambda} \Big] \E\big[\hat q^2\big].
\end{align*}

And, thus depending on the sign of $\Lambda$, maximizing $J^{\rm P}$ consists either on telling the truth or garbling information.
Solving for the threshold $b$ above which the TSO has incentive to provide no information yields the following proposition.
\begin{proposition} \label{Prop:truth2} It holds that
	\begin{enumerate}[\upshape (i)]
		\item The TSO has no incentive to deviate from telling the truth as long as
		\begin{align}\label{eq:condb}
			b \leq \bar b := \frac12 \Big( \frac32 \beta - c \Big) + \sqrt{ \frac14 \Big( \frac32 \beta - c\Big)^2 + \beta c }.
		\end{align}
		\item When $\bar b < b$,  it is optimal to provide no information.
	\end{enumerate}
\end{proposition}

\begin{remarkc}
	It is remarkable that the only factors that enter in the determination of the threshold $\bar b$ are the proportional part of the pricing rule and the cost of reduction effort of the consumer. The function $\bar b(c)$ is decreasing and satisfies $\bar b(0) =\frac32 \beta$ and $\lim_{c\to \infty} \bar b(c) = \beta$. Hence, whatever the value of $c$, there is a non-empty range of values for $b$, namely $(\beta, \bar b)$ where the TSO has an incentive to deceive the consumer. When reduction efforts are costless to the consumer, the threshold $\bar b$ equals $\frac32 \beta$, meaning that the TSO starts to take advantage of the signal only when she is paid 50\% more than the marginal cost. For instance, if the TSO always predicts high renewable energy generation, the consumer's actions is $\hat a(\tilde q)  =    a_{0,b} -  \lambda \hat q,$ where $\hat q$ is anticipated to be high. Thus, the consumer anticipating low prices makes little reduction effort. The revenue term of the TSO criteria, namely $p(y_0 - q- \hat a) (y_0-\hat a)$, outweighs the production cost term $f(y_0 - q - \hat a)$.
\end{remarkc}

\subsection{Multiple consumers}
We extend the former results in the case where the TSO faces a set of $n$ consumers. We assume she uses the same signal $\pi$ for all agents and send the same messages. The criterion of the TSO now becomes
\begin{align*}
	& \inf_{\pi} J^{{\rm T}}(\pi) := \E\big[ f( n\bar y_0 -  n \bar a -  n \bar q) \\
	& - p( n \bar y_0 - n \bar q -n  \bar a )\big(n \bar y_0 - n \bar a)\big) \big], \quad
	 p(x) = k + b x, \\
	 &\bar y_0 :=  \frac1n \sum_{j=1}^n y^0_j, \quad
	 \bar a :=  \frac1n \sum_{j=1}^n a_j, \quad \bar q := \frac1n Q.
\end{align*}
where the total renewable generation $Q$ is assumed to have a bounded support. Again, the TSO criterion can be rewritten as
\begin{align*}
	&J^{{\rm T}}(\pi) =  \E\big[ F_n(\bar q) \big] \\
	&+ n^2 \Big\{  (2b - \beta)  \bar y_0  \E\big[ \bar a \big]
	+  (\beta - b) \E\big[ \bar q \bar a\big] +   \Big(\frac12 \beta - b\Big) \E\big[\bar a^2\big] \Big\}.
\end{align*}
with $F_n(\bar q) :=  \frac12 n^2  \beta (\bar y_0 - \bar q)^2  - k n \bar q  - b n^2 ( \bar y_0 - \bar q) \bar y_0.$  Thus, the TSO only has to care about the average response of the consumers $\bar a$ to her signal $\pi$. On their side, consumers act as Bayesian Nash players sharing all the same prior on $\bar q$ as the TSO.  For each consumer, we set
\begin{align*}
	 U^i(\bar q, a_i,a_{-i}) := & u_i(y^i_{0}-a_i) - c_i(a_i) \\
	&- p(n \bar y_0 - n \bar q - n \bar a) (y_i^{0}-a_i).
\end{align*}
Each consumer wishes to maximise
\begin{align*}
	J^i(a_i,a_{-i}; \tilde q) := \E ^{\P}\Big[U^i(\bar q,a) | \tilde q \Big].
\end{align*}

\myblue{For a given signalling policy $\pi$, a Nash equilibrium of the consumers is vector $(a_i^{\ast,\pi})$ of functions of the message $\tilde q$ such that for $i=1,\ldots,n$ and for all $\tilde q$
\begin{align*}
	J^i(a_i^{\pi},a_{-i}^{\ast,\pi}; \tilde q) \leq J^i(a_i^{\ast,\pi},a_{-i}^{\ast,\pi}; \tilde q), \quad \text{for all} a_i^{\pi}.
\end{align*}}
We observe that the game played by consumers is an exact potential game (see \cite{Monderer96}). Define the potential function
\begin{align*}
\Phi(a) &:= \sum_{i=1}^n \big[ e_i a_i - \frac12 \Big( b + c_i \Big) a_i^2\big] - \frac{n}{2} b (\bar y_0 - \bar a - \hat{\bar q})^2, \\
e_i &:= k - u_i +  b  y^0_i.
\end{align*}
It holds that
\begin{align*}
\frac{\partial \Phi}{\partial a_i} (a) &= \frac{\partial U^i}{\partial a_i}(\bar q, a), \quad \text{for all}\quad i=1,\ldots,n.
\end{align*}
Moreover, $U^i(a_i,a_{-i})$ is strictly concave in $a_i$ for each $a_{-i}$. Hence the
Nash equilibria of the game with payoffs $U^i$ are equivalent to maximizers of the potential functions $\Phi$. As $\Phi$ is quadratic and negative definite, the equilibrium is unique and the Nash equilibrium is linear in $\hat{\bar q} := \E[\bar q|\tilde q]$. We have
\begin{align*}
	&  \bar a^\pi = {\rm const} - \hat \lambda_b \hat{\bar q}, \, \text{with}\, \hat \lambda_b := \frac{\hat c}{1+\hat c}, \quad
	\hat c := \frac1n  \sum_{i=1}^n \frac{b}{b+c_i}.
\end{align*}
Since the TSO's objective depends only on $\bar{a}$,  we are in the same setting as in the single case agent with a TSO facing one agent with average response $\bar a$. Hence, the condition of truth-telling of Proposition~\ref{Prop:truth2} still applies. Yet, the threshold condition for optimality of truth telling changes. Telling the truth is optimal as long as
\begin{align*}
h(b) := b \big( 1-\hat  \lambda_b\big) - \beta \Big( 1 - \frac12 \hat \lambda_b \Big) \leq 0.
\end{align*}
Because $\hat \lambda_b \in (0,\frac12)$, for $b$ sufficiently large, the constraint is violated and the TSO has incentive to provide no information. Uniqueness of the threshold is obtained by the monotonicity and continuity of the function $h$. But, unless some hypothesis are made on the $c_i$'s, it is unlikely to extract a condition as simple as~\eqref{eq:condb}. Nevertheless, if we denote by $b^\ast_n({\bf c})$ with ${\bf c} := (c_i)_{1\leq i \leq n}$, the threshold for $n$ players with a profile of reduction effort costs ${\bf c}$, it holds that
\begin{align}\label{eq:bastn}
b^\ast_1(\bar c) \leq b^\ast_n({\bf c}), \quad \text{with} \quad \bar c := \frac1n \sum_{i=1}^n c_i.
\end{align}
Relation~\eqref{eq:bastn} states that a population of $n$ consumers is less vulnerable to misinformation rather than a single representative consumer with the average characteristics of the population.


\section*{Acknowledgments}
We warmly thank one of the referees for pointing properties~\eqref{eq:ineq} and the fact that the consumers' game is a potential game.  The research in this paper was supported by the grant 6301-2 from the Indo-French Centre for the Promotion of Advanced Research. Ankur's work was also supported by the grant CRG/2019/002975 of the Science and Engineering Research Board, Department of Science and Technology, India. Ren\'{e}'s work was supported by the Chair EDF-CA CIB of Finance and Sustainable Development, the Finance for Energy Markets Research Initiative and the ANR project EcoREES ANR-19-CE05-0042.

%% file: proofs.tex
\appendix
\renewcommand*{\thesection}{\Alph{section}}
\section{Proof}

We establish inequality~\eqref{eq:bastn}. Note that
\begin{align*}
b^\ast_1(\bar c) = \frac{b}{b+\bar c} \frac{1}{1+\frac{b}{b+\bar c} } = \frac{b}{\bar c+2b},
\end{align*}
which is the threshold in the case of one representative consumer in Proposition~\ref{Prop:truth2} with cost of reduction $\bar c$. Denote
\begin{align*}
f(b) :=   \lambda_b \Big(\frac12 \beta - b \Big) + b - \beta, \quad \text{with} \lambda_b := \frac{b}{\bar c+2b}, \\
\hat f(b) :=   \hat \lambda_b \Big(\frac12 \beta - b\Big) + b - \beta, \quad \text{with} \hat \lambda_b := \frac{\hat c}{1+\hat c}.
\end{align*}
Both $f$ and $\hat f$ are monotonic non-decreasing functions. We have
\begin{align*}
f(b)-\hat f(b) =   \big( \lambda_b -\hat \lambda_b\big) \Big(\frac12 \beta - b \Big).
\end{align*}
We have $f(\frac12\beta) <0$ and $\hat f(\frac12\beta) <0$; thus, the thresholds are such that $\frac12 \beta < b$. Besides, by Jensen's inequality, it holds that
\begin{align*}
\frac{b}{b+\bar c} \leq \hat c = \frac1n \sum_{i=1}^n \frac{b}{b+c_i},
\end{align*}
and thus $\lambda_b < \hat \lambda_b$. Hence, for $b$ larger than $\frac12 \beta$, the difference $f(b)-\hat f(b)$ is positive and thus, the zero of $\hat f$ is larger than the zero of $f$.